\def\year{2020}\relax
\documentclass[letterpaper]{article} 
\usepackage{aaai20}  
\usepackage{times}  
\usepackage{helvet} 
\usepackage{courier}  
\usepackage[hyphens]{url}  
\usepackage{graphicx} 
\usepackage{graphicx}  
\usepackage[utf8]{inputenc} 
\usepackage[T1]{fontenc}    
\usepackage{url}            
\usepackage{booktabs}       
\usepackage{amsfonts}       
\usepackage{nicefrac}       
\usepackage{microtype}      
\usepackage{graphicx}
\usepackage{verbatim}
\usepackage{amsmath}
\usepackage{amsthm}
\usepackage{amssymb}
\usepackage{algorithm,algorithmicx}
\usepackage{algpseudocode}

\algnewcommand\algorithmicforeach{\textbf{for each}}
\algdef{S}[FOR]{ForEach}[1]{\algorithmicforeach\ #1\ \algorithmicdo}
\usepackage{xcolor}
\newtheorem{theorem}{Theorem}
 
\newtheorem{corollary}[theorem]{Corollary}

\newcommand{\ie}{\textit{i.e.}}
\newcommand{\eg}{\textit{e.g.}}

\DeclareMathOperator*{\argmax}{arg\,max}

\usepackage{multirow}
\usepackage{array}
\newcolumntype{P}[1]{>{\centering\arraybackslash}p{#1}}
\newcolumntype{M}[1]{>{\centering\arraybackslash}m{#1}}
\newcolumntype{L}[1]{>{\raggedright\let\newline\\\arraybackslash\hspace{0pt}}m{#1}}
\newcolumntype{C}[1]{>{\centering\let\newline\\\arraybackslash\hspace{0pt}}m{#1}}
\newcolumntype{R}[1]{>{\raggedleft\let\newline\\\arraybackslash\hspace{0pt}}m{#1}}
\title{Multi-Objective Multi-Agent Planning for Jointly Discovering and Tracking Mobile Objects}

\author{Hoa Van Nguyen,\textsuperscript{\rm 1} Hamid Rezatofighi,\textsuperscript{\rm 1} Ba-Ngu Vo,\textsuperscript{\rm 2} Damith C. Ranasinghe\textsuperscript{\rm 1} \\
\textsuperscript{\rm 1} The University of Adelaide \\
\textsuperscript{\rm 2} Curtin University\\
\{hoavan.nguyen,hamid.rezatofighi,damith.ranasinghe\}@adelaide.edu.au, ba-ngu.vo@curtin.edu.au
}

\begin{document}
	
	\maketitle
	
	\begin{abstract}
		We consider the challenging problem of online planning for a team of agents to autonomously search and track a time-varying number of mobile objects under the practical constraint of detection range limited onboard sensors.
		A standard POMDP with a value function that either encourages discovery or accurate tracking of mobile objects is inadequate to \textit{simultaneously} meet the \textit{conflicting goals} of searching for undiscovered mobile objects whilst keeping track of discovered objects. The planning problem is further complicated by misdetections or false detections of objects caused by range limited sensors and noise inherent to sensor measurements. We formulate a \textit{novel} multi-objective POMDP based on information theoretic criteria, and an online multi-object tracking filter for the problem. Since controlling multi-agent is a well known combinatorial optimization problem, assigning control actions to agents necessitates a greedy algorithm. We prove that our proposed multi-objective value function is a monotone submodular set function; consequently, the greedy algorithm can achieve a $(1-1/e)$ approximation for maximizing the submodular multi-objective function.
	\end{abstract}
	
	\section{Introduction}
	We study the problem of controlling a team of agents to jointly track discovered mobile objects and explore the environment to search for undiscovered mobile objects of interest. Such problems are ubiquitous in wildlife tracking~\cite{cliff2015online,kays2011tracking,hoa2019jofr,thomas2012wildlife}, search and rescue missions~\cite{gerasenko2001beacon,Murphy2008}. For instance, a team of unmanned aerial vehicles (UAVs) can be deployed to monitor activities of endangered radio-tagged wildlife in a survey scene, or to search for victims in a disaster response \cite{beck2018collaborative}. Hence, it is critical to not only search for undiscovered objects but also track the movements of discovered objects of interest. Consequently, the overall team's objectives arise as a natural multi-objective optimization problem, where several pertinent goals (\ie, tracking and discovering) need to be simultaneously achieved. 
    
    Intrinsically, searching for undiscovered objects whilst simultaneously tracking visible objects are \textit{competing} goals because, in practice, agent sensor systems, such as cameras, have limited detection range. A single agent may only observe a small region of space and a decision to leave a visible object to explore hitherto unseen regions will lead to losing track of visible objects. Therefore, an agent observing a small region of the search area needs to collaboratively interact with other agents to \textit{plan} its course of actions to collectively maximize the overall team's objectives of \textit{tracking} and \textit{discovering} multiple objects. 
    
    Multi-agent planning to achieve \textit{multiple competing objectives} remains a challenging problem because of the complex interactions between agents leading to combinatorial optimization problems~\cite{wai2018multi}. In practice, the problem is \textit{further} complicated because: \textit{i)} the agent sensors are not only limited in range but also sensitivity, and measurements are always subjected to environmental noise. Consequently, object detectors suffer from both missing detections of objects and false detections; and \textit{ii)} the number of objects of interest is often unknown, and varies with time since mobile objects can enter and leave the scene anytime~\cite{vo2012multi}.  Most critically, the computation of optimal planning actions must be timely for real-world applications.

	We propose a framework for multiple agents to jointly plan, search and track a time-varying number of objects using a novel multi-objective information-based value function formulation. Our multi-objective value function captures the competing objectives of planning for \textit{tracking} and \textit{discovery}. We adopt the random finite set (RFS) model for the collection of objects of interest to account for the random appearance and disappearance of objects and their dynamics. Our proposed multi-objective value function maximizes information gain over a look-ahead horizon for both discovered and undiscovered objects. Most importantly, our multi-objective value function is proven to be a monotone submodular set function; thus, we can cope with the intractability of the multi-objective optimization problem (MOP) by employing a greedy algorithm. Our ability to use a greedy algorithm facilitates the computation of approximately optimal control actions with linear complexity in the number of agents for realizing an online planning method.

	\textbf{Related Work:~}Multi-agent path planning in partially observable environments is a difficult problem
	for which the Partially Observable Markov Decision Processes (POMDP) approach has recently gained significant interest~\cite{David2010,MacDermed2013,Messias2011}. Although the cooperation problem can be formulated as a decentralized POMDP (Dec-POMDP), its exact solutions are NEXP-hard~\cite{Bernstein2002}. This is especially problematic for multi-agent POMDPs since the action and observation space grows exponentially with the number of agents~\cite{Amato2015}. To cope with this intractability, we adopt the MPOMDP centralized approach~\cite{Messias2011} for controlling multiple agents~\cite{dames2015autonomous,dames2017detecting,Reza2}.

    POMDP has also been employed for sensor selection problems, \eg, \cite{spaan2015decision,satsangi2018exploiting} proposed using the $\rho$POMDP~\cite{Araya2010nips} for a mobile agent to select $K$ in $N$ available sensors to search and track multiple objects. In particular, \cite{spaan2015decision} proposed a method that always assumes the existence of one extra object in the scene to encourage discovery. However, biasing the cardinality estimate generates sub optimal planing decisions at the cost of tracking performance. 
    
    Our study focuses on the problem of controlling a team of agents for the task of tracking and discovering mobile targets. The task requires a suitable tracking framework. Studies on tracking objects have employed approaches such as Multiple Hypotheses Tracking (MHT)~\cite{reid1979algorithm} or Joint Probabilistic Data Association (JPDA)~\cite{blackman1999design}. The complex nature of our problem requires a framework that has the notion of probability of a random collection due to a time-varying and random number of objects where the states of objects are random vectors. The random finite set (RFS) ~\cite{mahler2007statistical} is the only framework that has the notion of probability density of a random set. 
    Hence, we adopt RFS as our tracking framework.

	Information-based path planning under the RFS framework for a single agent has been studied in several works~\cite{beard2015void,hoang2014sensor,ristic2010sensor}. Most studies on multi-agent path planning using an RFS framework, is based on the Generalized Covariance Intersection (GCI) methods with the assumption that agents have a consensus view of all objects~\cite{gostar2016,Reza2} and using only a single look-ahead horizon. \cite{dames2017detecting} proposed to control multiple fixed-wing UAVs to localize mobile taxis with a single objective value function. For localizing and searching objects simultaneously, \cite{dames2015autonomous} and \cite{charrow2015active} considered a similar scenario, but only for stationary objects. Planning using multi-objective optimization (MOP) has not been explored yet, except for single sensor selection~\cite{zhu2019multi} or using the weighted sum method presented in ~\cite{charrow2015active} where the weighting parameters are difficult to define without prior knowledge. In contrast, we focus on optimizing all value functions (\ie, tracking and discovering) simultaneously using MOP. In particular, our proposed tracking and discovering value functions are based on information criteria. 
	The tracking value function maximizes the mutual information between future measurements and discovered object states under a multi-sensor Bernoulli filter; the discovering value function maximizes the mutual information between empty measurements and undiscovered object states under a grid occupancy filter.

	\textbf{Our contributions:}~The main contributions of our work are: \textit{(i)} We formulate a multi-agent planning problem with competing objectives and propose a planning algorithm for searching and tracking multiple mobile objects; \textit{(ii)} We unify tracking and planning algorithms under a Bernoulli-based model; \textit{(iii)} We prove that our proposed multi-objective value function is submodular; hence, the greedy algorithm can be used to rapidly determine the approximately optimal control actions with a bounded performance guarantee at $(1-1/e)$OPT.  
	
	\section{Problem Formulation}
    First, we introduce assumptions to help define our problem and introduce the notations we adopt in our work. Second, we provide a brief overview of the multi-sensor Bernoulli filter which unifies the tracking and discovering formulation. Next, we formulate our MPOMDP multi-agent planning approach for controlling the multi-agent team. 
    
    \subsection{Assumptions and Notations}
    \textbf{Assumptions:}~We consider a team of $S$ agents surveying a large area to detect and track an unknown and time-varying number of mobile objects using detection-based measurements.  We assume that each agent can localize itself (\eg, using an onboard GPS for UAVs) and that all agents can communicate to a central node to enables us to adopt the centralized approach for MPOMDP. Consequently, we assume that all of the measurements are transferred to a central node that analyzes received information and subsequently sends control actions to all of the agents. Here, we employ a discrete control action space to reduce the computational load~\cite{beard2015void,dames2017detecting}.  We further assume that the measurements from an object collected by the agents are conditionally independent given the object's state~\cite{charrow2015active,thrun2005probabilistic}. 
    
    \noindent\textbf{Notations:~}We use the convention that lower-case letters (\eg, $x$) represent single-object states, upper-case letters (\eg, $X$) represent multi-object (finite-set) states, and blackboard bold letters (\eg, $\mathbb{X,Z}$) represent spaces. We denote the inner product $\int f(x) g(x) dx = \langle f,g\rangle$.

	\subsection{Multi-sensor Bernoulli filter (MS-BF)}\label{sec:MSBF}
	In practice, an object can randomly enter and leave the surveillance region, hence the number of objects of interest is unknown and time-varying. Further, it is important to consider the existence of objects of interest to allow the agents to discover new objects when they enter the scene and to prevent agents following false-positives. This can be addressed by the random finite sets (RFSs) approach, first proposed by \cite{mahler2007statistical}. RFSs are finite-set valued random variables.  We assume that each measurement is uniquely identified, \eg, transmit frequencies from radio beacons \cite{cliff2015online,kays2011tracking,hoa2019jofr,hoa2019tsp,thomas2012wildlife} or MAC address \cite{beck2018collaborative,charrow2015active}, which is typical for wildlife tracking or search and rescue missions. Since each object is uniquely identified, we propose using a the multi-sensor Bernoulli filter (MS-BF)~\cite{vo2012multi}, where each object's state is a Bernoulli RFS, and run multiple MS-BF filters parallelly to track multiple objects. A Bernoulli RFS $X$ on $\mathbb{X}$ has at most one element with probability $r$ for being singleton or $1-r$ for being empty. Its probability density $\pi(\cdot) = (r,p(\cdot))$ given by
	
	$$
	\pi(X) = \begin{cases} 
	1 - r, & X = \emptyset, \\
	r \cdot p(x), & X = \{x\}. \\
	\end{cases} 
	$$

	\noindent\textbf{Object tracking with MS-BF:} We model each object's state at time $k$ by $X_k$ as a Bernoulli RFS. 
	The MS-BF propagates the two quantities: the existence probability $r$ and spatial density $p(\cdot)$. If the posterior density is $\pi_{k-1} = (r_{k-1}, p_{k-1}) $, then the predicted density $\pi_{k|k-1} = (r_{k|k-1}, p_{k|k-1})$ is also a Bernoulli RFS, with $r_{k|k-1} = r_{B,k}(1-r_{k-1}) + r_{k-1} \langle p_{S,k},p_{k-1} \rangle$; $p_{k-1}(x_k) = \big[r_{B,k}(1-r_{k-1}) b_{k}(x_k) + r_{k-1} \langle f_{k|k-1}(x_k|\cdot),p_{S,k}(\cdot)p_{k-1}(\cdot) \rangle  \big]/r_{k|k-1}$. Here, $r_{B,k}$ and $p_{S,k}$ are the probabilities of object birth and object survival, $b_k(\cdot)$ is the object birth density. 
	Further, the updated density $\pi_k$ is also a Bernoulli RFS, given by $\pi_k =  (r_k, p_k)  $ with $r_k = \Psi^{(S)}_k \circ \dots \circ \Psi^{(1)}_k (r_{k|k-1})$; $p_k = \Psi^{(S)}_k \circ \dots \circ \Psi^{(1)}_k (p_{k|k-1})$. Here, $\circ$ denotes composition (of operators),
	$\Psi^{(s)}_k$ is an update operator for agent $s$, \ie:
	\begin{align}
	[\Psi^{(s)}_k(r)] =& \langle \eta^{(s)}(Z^{(s)} | \cdot),p(\cdot) \rangle r/ \big[ (1-r) \notag e^{-\lambda^{(s)}}  \\&+ r \langle \eta^{(s)}(Z^{(s)} | \cdot),p(\cdot) \rangle \big],\label{eq_psi_r_track} \\ 
	[\Psi^{(s)}_k (p)](x) =& \eta^{(s)}(Z^{(s)} | x)p(x)/\langle \eta^{(s)}(Z^{(s)} | \cdot),p(\cdot) \rangle\label{eq_psi_p_track}
	\end{align}
	where the superscript $(s)$ denotes the parameters of agent $s$, 
	$\lambda^{(s)}$ is the clutter rate, and $\eta^{(s)}(Z^{(s)} |x) $ denotes the likelihood of measurement set $Z^{(s)}$ from agent $s$ given the object's state $x$. $\eta^{(s)}(Z^{(s)} |x)$ is also a Bernoulli RFS, given by
	$
	\eta^{(s)}(Z^{(s)} |x) = \begin{cases} 
	1 - p_d^{(s)}(x), & Z^{(s)}  = \emptyset, \\
	p_d^{(s)}(x) g^{(s)}(z|x), & Z^{(s)}  = \{z\}. \\
	\end{cases} 
	$
	
	\noindent Here, $p_d^{(s)}(x)$ is the probability that agent $s$ detects object $x$, and $g^{(s)}(z|x)$ is the (conventional) likelihood function of measurement $z$ given object's state $x$. 
	\subsection{Planning}
	At time $k$, the team of $S$ agents needs to plan how they manoeuvre over the time interval $k+1:k+H$ to improve its estimation of the states of multiple objects $X_k$, where $H$ denotes the look-ahead horizon length. Let $\mathbb{A} \subseteq \mathbb{R}^N$ be all possible set of control actions for a given agent. 
	When the control action $a^i_k \in \mathbb{A}$ is applied to an agent $i$, it follows a trajectory comprised of sequence of  the discrete poses $u^i_{k+1:k+H}(a^i_k) = [u^i_{k+1},\dots,u^i_{k+H}]^T$ with corresponding measurements 
	$Z^i_{k+1:k+H}(a^i_k) = [Z^i_{k+1},\dots,Z^i_{k+H}]^T$ (for notational compactness, we omit the dependence on $X_k$ here).
	Let 
	$A_k = [a^1_k,\dots,a^S_k]^T \in \mathbb{A}^S $ be the control actions where $\mathbb{A}^S = \mathbb{A} \times \dots \times \mathbb{A}$ is the control action space for $S$ agents, and the corresponding measurement set is $Z_{k+1:k+H}(A_k) = [Z^1_{k+1:k+H}(a^1_k),\dots,Z^S_{k+1:k+H}(a^S_k)]^T$. 
	
	The objective of path planning is to find the optimal action $A^{*}_k \in \mathbb{A}^S$ that maximizes the value function, \ie, 
	\begin{align} \label{eq_opt_act}
	    A^{*}_k = \argmax_{A_k \in \mathbb{A}^S} V(X_{k+1:k+H},Z_{k+1:k+H}(A_k)).
	\end{align}
	where $V(X_{k+1:k+H},Z_{k+1:k+H}(A_k)) = \mathbf{E}\Big[\sum_{j=1}^H \mathcal{R}(X_{k+j},Z_{k+j}(A_k))\Big]$ is the value function or the expected sum of immediate rewards $\mathcal{R}(\cdot)$ over a finite horizon $H$. 
	Since an analytic solution does not exist for the expected reward, we use the predicted ideal measurement set (PIMS) \cite{mahler2004multitarget}---a computationally low-cost approach. The value function is calculated:
	\begin{align}
	V(X_{k+1:k+H},\hat{Z}_{k+1:k+H}(A_k)) = \sum_{j=1}^H \mathcal{R}(X_{k+j},\hat{Z}_{k+j}(A_k)) \notag
	\end{align}
	where $\hat{Z}_{k+j}(A_k)$ denotes the ideal measurement set of $Z_{k+j}(A_k)$ calculated using  the measurement model and the estimated states of objects without measurement noise. For notational compactness, we write the value function $V(X_{k+1:k+H},\hat{Z}_{k+1:k+H}(A_k))$ as $V(A_k)$. 
	
	\section{Planning for Tracking and Discovering Multiple Objects}
	
	\subsection{Planning for tracking discovered mobile objects}
	In this problem, we consider maximizing an information-based reward function to reduce the overall uncertainty of the discovered mobile objects because more information naturally implies less uncertainty. In particular, we propose using the mutual information $I(X;Z)$ between the object's state $X$ and measurement state $Z$ as the immediate reward function, and the long-term sum of rewards over a finite horizon $H$, or so-called the value function is given by
	\begin{align}\label{eq_R1_Track}
	V_1(A_k) = \sum_{j=1}^H I(X_{k+j};\hat{Z}_{k+j}(A_k))
	\end{align}
	where 
	$I(X;Z)= h(X) - h(X|Z)$, with $h(X)$ is the generalization of differential entropy for a finite set $X \subseteq \mathbb{X}$ with density $f(X)$ defined as
	$ h(X) = -\int_\mathbb{X} f(X) \log f(X) \delta X $; here $\int_\mathbb{X} \cdot \delta X$ is the set integral~\cite{mahler2007statistical}. For Bernoulli RFS, this integration is simplified to $ h(X) = -\big[ f(X=\emptyset) \log f(X=\emptyset) + \int f(X=x) \log f(X=x) dx\big] $. 
	We have the following theorem.  
	\begin{theorem} \label{eq_theorem_MI}
		The mutual information $I(X;Z)$ between the object state $X$ and measurement state $Z$ is a monotone submodular set function of $Z$. 
	\end{theorem}
	\begin{proof}
		We want to prove that this mutual information $I(X;Z)$ is a monotone submodular set function, \ie, for $Z_1 \subseteq Z_2 \subseteq \mathbb{Z}$, and $z \in \mathbb{Z} \setminus Z_2$ independent of $Z_1$ and $Z_2$:
		\begin{align}
		I(X;Z_2,\{z\} ) - I(X;Z_2) &\leq 	I(X;Z_1,\{z\} ) - I(X;Z_1).\notag
		\end{align}
		Since $Z_1 \subseteq Z_2 \subseteq \mathbb{Z} $, using mutual information inequalities \cite[p.50]{cover2012elements}, we have:
		\begin{align}
		I(Z_2;\{z\}) &\geq  I(Z_1;\{z\}), \notag \\
		\Leftrightarrow h({z}) - h({z}|Z_2) &\geq h({z}) - h({z}|Z_1), \notag\\
		\Leftrightarrow h({z}|Z_1) &\geq h({z}|Z_2), \notag\\
		\Leftrightarrow h(Z_1,\{z\}) - h(Z_1)  &\geq  h(Z_2,\{z\}) - h(Z_2). \label{eq_18}
		\end{align}
		Further, since $I(Z_2;\{z\}|X) =I(Z_1;\{z\}|X)= 0$ is due to $z$ is independent of $Z_1$ and $Z_2$ given $X$, we have:
		\begin{align}
		h(\{z\}|X) &= h(\{z\}|X,Z_2) + I(Z_2;\{z\}|X) = h(\{z\}|X,Z_2) \notag\\ &= h(X,Z_2,\{z\}) - h(X,Z_2)\notag, \\
		h(\{z\}|X) &= h(\{z\}|X,Z_1) +  I(Z_1;\{z\}|X) \notag\\ &= h(X,Z_1,\{z\}) - h(X,Z_1)\notag.
		\end{align}
		Hence,
		\begin{align}
		h(X,Z_2,\{z\}) - h(X,Z_2) = h(X,Z_1,\{z\}) - h(X,Z_1). \label{eq_21}
		\end{align}
		
		Subtracting \eqref{eq_18} from \eqref{eq_21}, we have:
		\begin{align} 
		&[h(X,Z_2,\{z\}) - h(X,Z_2) ] - [h(Z_2,\{z\}) - h(Z_2) ] \notag\\ &\geq [h(X,Z_1,\{z\}) - h(X,Z_1) ] - [h(Z_1,\{z\}) - h(Z_1) ]   \notag
		\end{align}
		Using differential entropy chain rules \cite[p.253]{cover2012elements}, we have that $h(X|Z_2,\{z\}) = h(X,Z_2,\{z\}) - h(Z_2,\{z\}) $ and $h(X|Z_2) = h(X,Z_2) - h(Z_2)$, thus the above equation is equivalent to
		\begin{align}
		&h(X|Z_2,\{z\}) - h(X|Z_2) \geq h(X|Z_1,\{z\}) - h(X|Z_1)\notag \\
		\Leftrightarrow &[h(X) - h(X|Z_2,\{z\})] - [h(X) - h(X|Z_2)] \notag\\  &\leq [h(X)-h(X|Z_1,\{z\}] - [h(X)-h(X|Z_1)], \notag\\
		\Leftrightarrow 	&I(X;Z_2,\{z\} ) - I(X;Z_2) \leq 	I(X;Z_1,\{z\} ) - I(X;Z_1). \notag
		\end{align}
		Thus, $I(X;Z)$ is a submodular set function. 
		Further, using the chain rule we have: 
		\begin{align}
		I(X;Z_2,\{z\}) - I(X;Z_2) =  I(X;Z_2|\{z\}) \geq 0 \notag
		\end{align}
		Therefore, $I(X;Z)$ is a monotone submodular set function. 
	\end{proof}
	
	\noindent\textbf{Remark:} Our mutual information formulation is different to that in \cite{krause2008near} used for sensor selection problems. Krause et al. showed that for $Z \subseteq \mathbb{Z}$, the mutual information $I(Z;\mathbb{Z}\setminus Z)$ is a submodular set function. In other words, the mutual information $I(Z_1;Z_2)$ is submodular with the \textit{{property}} that $Z_1 \cup Z_2 = \mathbb{Z}$ and $|\mathbb{Z}|$ is fixed. In contrast, we measure the mutual information between the random set object state $X$ and the random set measurement state $Z$ and prove $I(X;Z)$ is also a submodular set function of $Z$ without the aforementioned \textit{property}.  
	\begin{corollary} The value function
		$V_1(A_k)$ in \eqref{eq_R1_Track} is a monotone submodular set function.
	\end{corollary}
	\begin{proof}
		Since $I(X_k;\hat{Z}_{k+j}(A_k))$ is a monotone submodular set function and $V_1(A_k)$ is a positive linear combination of it, according to \cite[p.272]{nemhauser1978analysis}, $V_1(A_k)$ is a monotone submodular set function.
	\end{proof}
	
	\noindent\textbf{Mutual Information Calculation based on MS-BF}: Assume that each object $i$ is associated with a Bernoulli distribution $\pi(X_i) =  (r_i, p_i)$.
	Let $p_i(x)$ be approximated by a set of $N_s$ particles, such that $p_i(x) \approx \sum\limits_{m=1}^{N_s} w_i^{(m)} \delta (x^{(m)} -x) $ with $ \sum\limits_{m=1}^{N_s} w_i^{(m)} = 1$ and $\delta(\cdot)$ is the Kronecker delta function; $X = X_1 \cup \dots \cup X_n$ be the state of multiple objects. 
	Since each object is uniquely identified by its label and estimated by an individual Bernoulli filter, we have
	\begin{align}\label{eq_diff_entropy_def}
	h(X) = \sum_{i}^{n} h(X_i) \approx &\sum_{i}^{n} \Big[- (1 - r_i) \log(1-r_i)  \\ 
	&- r_i \sum\limits_{m=1}^{N_s} \big[w_i^{(m)} \log(r_iw_i^{(m)})\big]\Big].\notag
	\end{align}
	According to the definition of the mutual information $I(X;Z) = h(X) - h(X|Z)$, thus the tracking value function $V_1(A_k)$ can be calculated as $V_1(A_k) = \sum_{j=1}^H \big[ h(X_{k+j}) - h(X_{k+j}|\hat{Z}_{k+j}(A_k))\big]$, where $h(X_{k+1})$ is calculated directly in \eqref{eq_diff_entropy_def}. For $h(X_{k+j}|\hat{Z}_{k+j}(A_k))$, it has the same form as in \eqref{eq_diff_entropy_def}; however, $r_{k+j,i}$ and $w_{k+j,i}^{(m)}$ are calculated by propagating $r_{k,i}$ and $w_{k,i}^{(m)}$ from time $k$ to $k+j$ using \eqref{eq_psi_r_track} and \eqref{eq_psi_p_track} respectively with the ideal measurements $\hat{Z}_{k+j}(A_k)$.
	
	\subsection{Planning to search for undiscovered mobile objects}
	\noindent\textbf{Occupancy Grid Filter:} Since an agent is equipped with a sensor with a limited detection range, we propose using an occupancy grid to represent the probability of any undiscovered objects \cite{elfes1989using}. We extend the static grid approach in \cite{charrow2015active,thrun2005probabilistic} by incorporating the birth probability into each occupancy cell to account for the possibilities of mobile objects entering and leaving the survey area, anytime. The survey area is divided into an  occupancy grid $G = \{g^1,\dots,g^{N_g}\} \subset \mathbb{R}^N $, where each cell $g^i \in G$ associated with a Bernoulli random variable $r^i$. Here, $r^i$ is the probability that cell $g^i$ contains at least one undiscovered object. 
	For initialization, we set $r^{i}_{0} = r_B$ such that every cell has the same prior. Each cell $i$ is propagated through MS-BF over time using the predict and update equations. In particular, let $r^i_{k-1}$ be the probability of cell $g^i$ containing at least one undiscovered object, then its predict and update probabilities at time $k$ are \eqref{eq_occu_grid_r_pred} and \eqref{eq_occu_grid_r_update}. Note that since these objects are yet to be discovered, we use empty measurements for all agents (denoted as $Z^{\emptyset}$ to update). 
	\begin{align}
	    r^{i}_{k|k-1} &= r_B(1-r^i_{k-1}) + r^i_{k-1}p_S, \label{eq_occu_grid_r_pred}\\
	    r^i_k &= \Psi^{(S)}_k \circ \dots \circ \Psi^{(1)}_k (r^{i}_{k|k-1}), \label{eq_occu_grid_r_update}
	\end{align}
	where $[\Psi_k^{(s)}(r^i)] = (1-p_d^{(s)}(g^i))r^i/\big[1-r^i + r^i(1-p_d^{(s)}(g^i))\big].$
	
	\vspace{0.2cm}
	\noindent\textbf{Searching for undiscovered objects:} As before, we propose using mutual information as the immediate reward function. We want to maximize the mutual information between the estimated occupancy grid $G$ and the ideal empty future measurement $\hat{Z}^{\emptyset}_{k+1:k+H}(A_k)$, \ie,
	\begin{align}\label{eq_R2_discover}
	V_2(A_k) = \sum_{j=1}^H I(G_{k};\hat{Z}^{\emptyset}_{k+j}(A_k)),
	\end{align}
	where $I(G_{k};\hat{Z}^{\emptyset}_{k+j}(A_k)) = \mathcal{H}(G_{k}) - \mathcal{H}(G_{k}|\hat{Z}^{\emptyset}_{k+j}(A_k))$ and $\mathcal{H}(G_{k})$ is the Shannon entropy of $G_{k}$:
	\begin{align}\label{eq_shannon_entropy_def}
	\mathcal{H}(G_{k}) = -\sum_{i=1}^{N_g} \big[&r_{k}^i \log(r_{k}^i)  \\
	&+(1-r_{k}^i) \log(1-r_{k}^i)\big], \notag
	\end{align}
	and $\mathcal{H}(G_{k+j}|\hat{Z}^{\emptyset}_{k+j}(A_k))$ has the same form as in  \eqref{eq_shannon_entropy_def} with $r^i_{k+j}$ is calculated by propagating $r^i_{k+j}$ from $k$ to $k+j$ using the update step in \eqref{eq_occu_grid_r_update} with empty measurements $\hat{Z}^{\emptyset}_{k+j}(A_k)$.
	\begin{theorem} The value function
		$V_2$ in \eqref{eq_R2_discover} is a monotone submodular set function.
	\end{theorem}
	\begin{proof} We can apply a similar strategy as per Theorem \ref{eq_theorem_MI} to prove that $I(G_k;\hat{Z}^{\emptyset}_{k+j}(A_k))$ is a monotone submodular set function, note that $\mathcal{H}(\cdot)$ is the Shannon entropy (a discrete version of differential entropy $h(\cdot)$). Further, since $V_2(A_k)$ is a positive linear combination of $I(G_k;\hat{Z}^{\emptyset}_{k+j}(A_k))$, according to \cite[p.272]{nemhauser1978analysis}, $V_2(A_k)$ is a monotone submodular set function.
	\end{proof}
	
	\subsection{Multi-objective value function for tracking and discovering}
	In this problem, we want to control the team of agents to perform both tracking  and discovering; this naturally leads to a multi-objective problem.  Specifically, we want to maximize $$V(A_k) = [V_1(A_k), V_2(A_k)]^T$$ subject to $A_k \in \mathbb{A}^S$ where $V_1$ and $V_2$ are defined in \eqref{eq_R1_Track} and \eqref{eq_R2_discover}, respectively. Multi-objective optimization provides a meaningful notion of multi-objective optimality such as the Pareto-set, which represents trade-offs between the objectives such that there is no other solution that can improve one objective without degrading any remaining objectives~\cite{Whiteson2016}.
	Online planning necessitates selecting one compromised solution from the Pareto-set on-the-fly. One approach is Robust Submodular Observation Selection (RSOS) \cite{krause2008robust}, which is robust against the worst possible objective; however, even if each $V_i$ is submodular, $V_{min} = \min_{i}V_i$ is generally not submodular. Other approaches include Weighted Sum (WS) and Global Criterion Method (GCM); simplicity of these methods are not only attractive for meeting the demands of online planning but also result in a submodular value function. In this work, we adopt GCM 
	to select the compromised solution considering the distance from the ideal solution. 
	Inspired by \cite{Koski1993}, we define the value function $V_{mo}$ (with $V_{mo}(\emptyset) = 0$) as:
	\begin{align}\label{eq_comb_reward}
	V_{mo}(A_k) =  \sum_{i=1}^2 \dfrac{V_i(A_k) - \min\limits_{A_k \in \mathbb{A}^S}V_i(A_k)}{\max\limits_{A \in \mathbb{A}^S}V_i(A_k) - \min\limits_{A_k \in \mathbb{A}^S}V_i(A_k)}.
	\end{align}

	The global criterion method admits a unique optimal solution from the Pareto-set~\cite{coello2007evolutionary}. Hence, the multi-objective problem becomes
	\begin{align}\label{eq_optimal_action}
	A^{*}_k = \argmax_{A_k \in \mathbb{A}^S} V_{mo}(A_k).
	\end{align}
	Since finding the optimal control action $A^* \in \mathbb{A}^S$ is a combinatorial optimization problem, we want to show that the multi-objective value function $V_{mo}(A)$ in \eqref{eq_comb_reward} is also a monotone submodular set function on $Z$. This enables us to use the greedy algorithm to find the optimal action that approximately maximize this multi-objective value function. 
	
	\begin{corollary}
		The multi-objective value function $V_{mo}$ in \eqref{eq_comb_reward} is a monotone submodular set function. 
	\end{corollary}
	\begin{proof}
		Since $V_i(A_k))$ is a monotone submodular set function and $V_{mo}(A_k)$ is a positive linear combination of it, according to \cite[p.272]{nemhauser1978analysis}, $V_{mo}(A_k)$ is a monotone submodular set function.
	\end{proof}
	
	\subsection{Greedy search algorithm}
	
	\begin{figure*}[!tb]
		\centering
		\includegraphics[width=1.8\columnwidth]{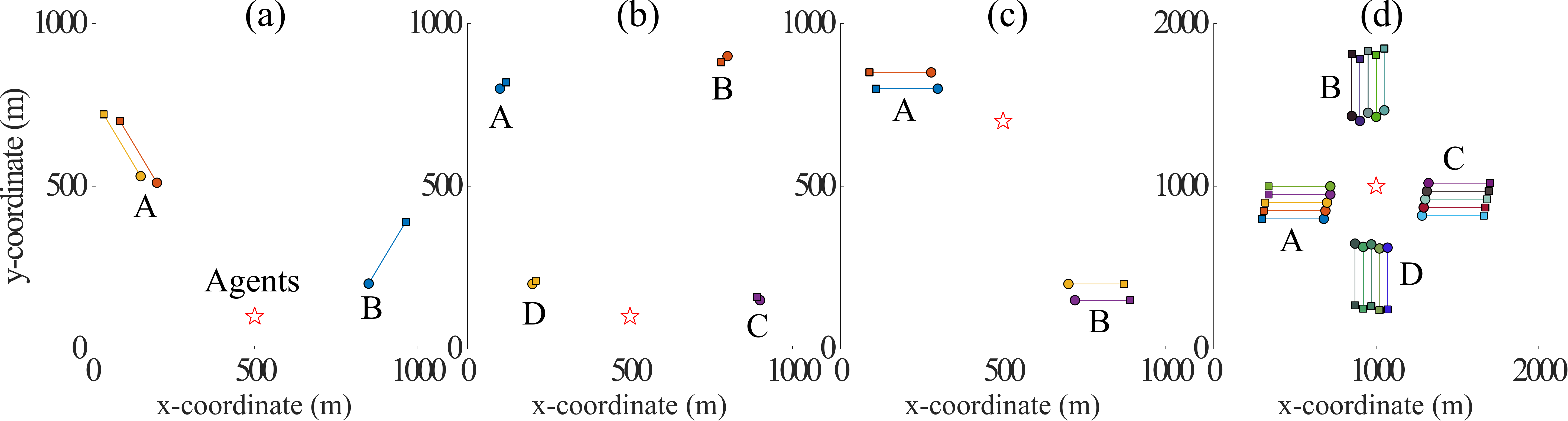}
		\caption{Setup for four scenarios: (a) Scenario 1; (b) Scenario 2; (c) Scenario 3; (d) Scenario 4. Start/Stop locations for each object are denoted by $\bigcirc/\Box$. Start locations for agents are denoted by $\star$. }
		\label{fig_scenario_setting}
	\end{figure*}

	We proved our multi-objective value function $V_{mo}(\cdot)$ is a monotone submodular set function---see \textbf{Corollary 4}. For submodular functions, \cite{nemhauser1978analysis} proved the greedy search algorithm guarantees a performance bound at $(1-1/e)$OPT, where OPT is the optimal value of the submodular function. Therefore, if the optimal value of our value function is $V_{mo}(A^*)$, we can simply state the following fundamental performance bound for our submodular value function: 
	\begin{theorem} \label{theorem_lower_bound} From \cite{nemhauser1978analysis}. Let $A_G$ be the output greedy control action  and $A^*$ be the optimal control action evaluated using brute-force method of \eqref{eq_optimal_action}. Then
		\begin{align}
		V_{mo}(A_G) \geq  (1-1/e) V_{mo}(A^*)\label{eq_theorem_lower_bound}
		\end{align}
		where $e = 2.718\dots $ is the base of the natural logarithm.
	\end{theorem}
	Hence, we propose using the greedy search algorithm by simply adding agents sequentially and picking the next agent which provides the maximum value function $V_{mo}(\cdot)$ as presented in Algorithm \ref{greedy_algo}. 
	
	\begin{algorithm}[!h]
		\caption{Greedy algorithm}\label{greedy_algo}
		\begin{algorithmic}[1] 
			\State \textbf{Input}: $V_{mo}(\cdot), \mathbb{A}$ \Comment{Value function and the action space}
			\State \textbf{Output}: $A_G \in \mathbb{A}^S $ \Comment{Greedy control actions for all agents}
			\State $A_G \gets \emptyset$ \Comment{Initialize the greedy control action}
			\State $P \gets \emptyset$ \Comment{Initialize the agent planned list}
			\State $U \gets \{1,\dots,S \}$ \Comment{Initialize list of agents to plan}
			\While{$U\not=\emptyset$} 
			\ForEach{ $s \in U$}
			\State $A^{s},V^s_c= \argmax\limits_{A \in \mathbb{A}^{V\cup \{s\}}} V_{mo}(A)$ \Comment{Find the best action and value function for each agent in $U$}
			\EndFor
			\State $s^*= \argmax\limits_{s \in U} V^s_c$ \Comment{Select the agent $s^*$ that provides the best value function}
			\State $A_G \gets A_G \cup \{A^{s^*}\} $ \Comment{Save the greedy control action for agent $s^*$}
			\State $P \gets P \cup \{s^*\} $ \Comment{Add agent $s^*$ into the planned list}
			\State $U \gets U \setminus \{s^*\} $ \Comment{Remove agent $s^*$ from the list of agents to plan}
			\EndWhile\label{euclidendwhile}
			\State \textbf{return} $A_G$
		\end{algorithmic}
	\end{algorithm}
	
    \section{Experiments}
    We evaluate the proposed value function using a series of comprehensive synthetic experiments since we can control all of the parameters of the problem, especially with a time-varying number of agents and objects. 
    We compare three planning algorithm formulations: \textbf{\textit{(i)}} using the single objective value function $V_1(\cdot)$ in \eqref{eq_R1_Track} for tracking. \textbf{\textit{(ii)}} using a single objective value function based on our \textit{new} discovery value function $V_2(\cdot)$ in \eqref{eq_R2_discover}. \textbf{\textit{(iii)}} using our proposed multi-objective value function $V_{mo}(\cdot)$.

    We use optimal sub-pattern assignment (OSPA)~\cite{schuhmacher2008consistent} to measure performance. 
    We report~\textbf{OSPA Dist} as the \textit{\textbf{main}} metric to evaluate the \textit{overall} team performance since it incorporates both tracking and discovery indicators. For further insights into our planning formulations, we also report:~\textbf{\textit{(i)}}~\textbf{OSPA Loc} as a localization accuracy measure, \textbf{\textit{(ii)}} \textbf{OSPA Card} as an object discovery performance measure; and \textbf{\textit{(iii)}} \textbf{Search Area Entropy} as the average entropy of the occupancy grid to measure the coverage area of the team. For demonstration, a team of quad-copter UAVs flying at different altitudes is considered. The detailed parameter settings are provided in the appendix, while scenario setups are shown in Figure~\ref{fig_scenario_setting}. Our experiments considered four different scenarios and two different detection-based sensors subject to noisy measurements.

    \textbf{Scenario 1 (FastMoving)}:~\textit{Three fast moving objects in two groups travelling in the same direction.}
    A team of agents  starts at  $[500,100]^T$~m as depicted in Figure~\ref{fig_scenario_setting}a. 
    
    \textbf{Scenario 2 (LateBirth)}:~\textit{Late birth objects}.
    We investigate a searching and tracking scenario in Figure~\ref{fig_scenario_setting}b) with \textit{four} slow-moving mobile objects using a team of agents. Here, the groups of objects $D$ and $C$ enter the scene when the agents are out of their detection range---late birth. This scenario \textit{favours} agent planning with the discovery value function encouraging exploration and demonstrates the effectiveness of our multi-objective value function with its competing tracking and discovery objectives. 
    
    \textbf{Scenario 3 (Opposite)}:~\textit{Four objects in two groups (A and B) moving rapidly in opposing directions}. Figure~\ref{fig_scenario_setting}c illustrates the scenario. We use this setting to confirm the effectiveness of our multi-objective value function. Now, the possibility to discover group $B$ out of the sensor detection range must be achieved through exploration while planning to track group $A$ in the vicinity of the agents is immediately rewarded by the tracking objective.

	\begin{table*}[tb!]
		\centering
		\caption{Comparing multi-agent planning for tracking mobile objects using our multi-objective value function $V_{mo}$ across Scenarios 1, 2, 3 and 4 with detection range $r_d = 200$~m. $V_1$ and $V_2$ are baselines and the results are averaged over 20 MC trials.}
		\label{tab_scenario1234}
		\resizebox{2.0\columnwidth}{!}{%
			\begin{tabular}{@{}|c|c|r|r|r|r|r|r|r|r|@{}}
				\toprule
				&  & \multicolumn{4}{c|}{\textbf{Scenario 1 (FastMoving)}} & \multicolumn{4}{c|}{\textbf{Scenario 2 (LateBirth)}} \\ \midrule
				& Indicators & \multicolumn{1}{c|}{\begin{tabular}[c]{@{}c@{}}Overall\\ Performance\end{tabular}} & \multicolumn{1}{c|}{\begin{tabular}[c]{@{}c@{}}Tracking\\ Performance\end{tabular}} & \multicolumn{2}{c|}{\begin{tabular}[c]{@{}c@{}}Discovery \\ Performance\end{tabular}} & \multicolumn{1}{c|}{\begin{tabular}[c]{@{}c@{}}Overall\\ Performance\end{tabular}} & \multicolumn{1}{c|}{\begin{tabular}[c]{@{}c@{}}Tracking\\ Performance\end{tabular}} & \multicolumn{2}{c|}{\begin{tabular}[c]{@{}c@{}}Discovery\\ Performance\end{tabular}} \\ \midrule
				\textit{Agents} & \textit{\begin{tabular}[c]{@{}c@{}}Value \\ Functions\end{tabular}} & \multicolumn{1}{c|}{\textbf{\begin{tabular}[c]{@{}c@{}}OSPA \\ Dist (m)\end{tabular}}} & \multicolumn{1}{c|}{\textbf{\begin{tabular}[c]{@{}c@{}}OSPA \\ Loc (m)\end{tabular}}} & \multicolumn{1}{c|}{\textbf{\begin{tabular}[c]{@{}c@{}}OSPA \\ Card  (m)\end{tabular}}} & \multicolumn{1}{c|}{\textbf{\begin{tabular}[c]{@{}c@{}}Search Area \\ Entropy (nats)\end{tabular}}} & \multicolumn{1}{c|}{\textbf{\begin{tabular}[c]{@{}c@{}}OSPA \\ Dist (m)\end{tabular}}} & \multicolumn{1}{c|}{\textbf{\begin{tabular}[c]{@{}c@{}}OSPA \\ Loc (m)\end{tabular}}} & \multicolumn{1}{c|}{\textbf{\begin{tabular}[c]{@{}c@{}}OSPA \\ Card (m)\end{tabular}}} & \multicolumn{1}{c|}{\textbf{\begin{tabular}[c]{@{}c@{}}Search Area\\ Entropy (nats)\end{tabular}}} \\ \midrule
				\multirow{3}{*}{$S= 3$} & $V_1$ & 33.9 & 4.4 & 29.5 & 0.23 & 57.0 & 4.0 & 53.0 & 0.22 \\
				& $V_2$ & 21.2 & 9.7 & 11.5 & 0.12 & \textbf{41.1} & 10.3 & 30.8 & 0.12 \\
				& $V_{mo}$ & \textbf{17.7} & 6.1 & 11.6 & 0.17 & 52.1 & 5.2 & 46.9 & 0.17 \\ \midrule
				\multirow{3}{*}{$S = 5$} & $V_1$ & 25.4 & 5.1 & 20.3 & 0.2 & 53.4 & 3.6 & 49.8 & 0.17 \\
				& $V_2$ & 20.3 & 9.2 & 11.1 & 0.09 & 43.9 & 9.5 & 34.4 & 0.09 \\
				& $V_{mo}$ & \textbf{16.8} & 5.7 & 11.1 & 0.13 & \textbf{38.8} & 5.1 & 33.7 & 0.11 \\ \midrule
				&  & \multicolumn{4}{c|}{\textbf{Scenario 3 (Opposite)}} & \multicolumn{4}{c|}{\textbf{Scenario 4 (Explosion)}} \\ \midrule
				\multirow{3}{*}{$S= 3$} & $V_1$ & 51.7 & 3.0 & 48.7 & 0.24 & 53.7 & 3.1 & 50.6 & 0.32 \\
				& $V_2$ & 18.3 & 12.8 & 5.5 & 0.12 & 55.0 & 29.8 & 25.2 & 0.28 \\
				& $V_{mo}$ & \textbf{11.1} & 5.9 & 5.2 & 0.18 & \textbf{40.7} & 9.7 & 31.0 & 0.32 \\ \midrule
				\multirow{3}{*}{$S = 5$} & $V_1$ & 51.2 & 2.9 & 48.3 & 0.24 & 36.9 & 5.0 & 31.9 & 0.30 \\
				& $V_2$ & \textbf{10.5} & 6.9 & 3.6 & 0.09 & 31.2 & 19.5 & 11.7 & 0.25 \\
				& $V_{mo}$ & \textbf{10.5} & 5.9 & 4.1 & 0.15 & \textbf{17.4} & 6.4 & 11.1 & 0.29 \\ \bottomrule
			\end{tabular}%
		}

	\end{table*}
	
	\begin{figure}[tb!]
		\centering
		\includegraphics[width=0.8\columnwidth]{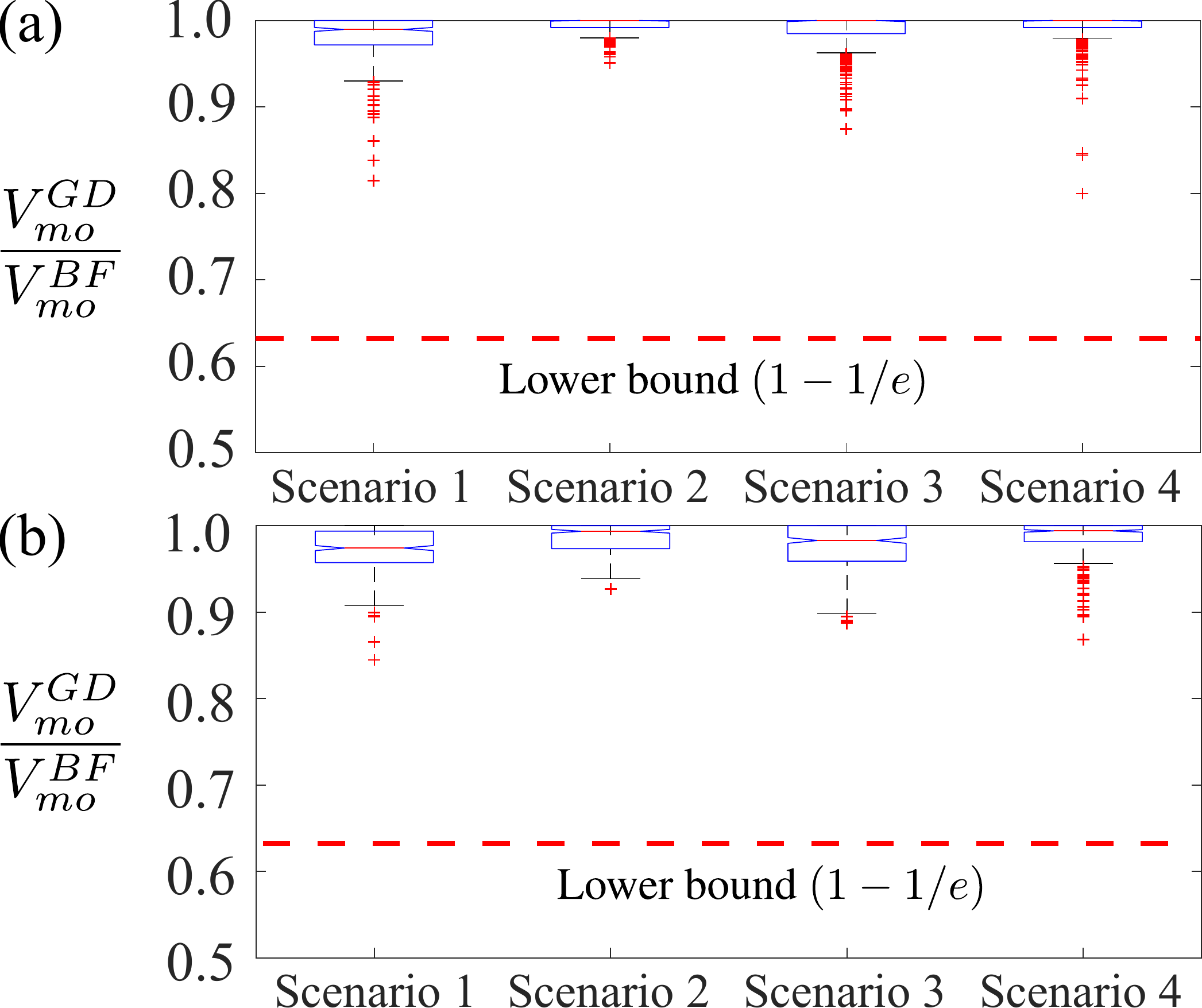} 
		\caption{Multi-objective value function ratio between the greedy $V_{mo}^{GD}$  and brute-force $V_{mo}^{BF}$ algorithms with agents (a) $S = 2$ and (b) $S=3$ (20 MC runs, range and bearing based sensor with $r_d=200$~m).}
		\label{fig_reward_ratio}
	\end{figure}
	
	\textbf{Scenario 4 (Explosion)}:~\textit{Multiple groups of fast moving objects in opposing directions}. Here, we consider a team of agents to search and track \textit{20 fast moving mobile objects} as shown in  Figure~\ref{fig_scenario_setting}d.
    
    \textbf{Detection-based sensors:} \textbf{\textit{(i)}} We considered agents equipped with a range and bearing based sensor---common in wildlife tracking~\cite{cliff2015online} for example. Let $u^s = [p^s_x,p^s_y,p^s_z]^T$ be the position of agent $s$, $x_p = [p_x,p_y,1]^T$ be the position of object $x$, each detected object $x$ leads to a noisy measurement $z$ of range and bearing given by: $z = \big[\arctan{[(p_y - p_y^s)/(p_x -p_x^s)]},||x_p-u^s||\big]^T + v$. Here,$|| \cdot ||$ is the Euclidean norm; $v \sim \mathcal{N}(0,R)$ with $R = \mathrm{diag}(\sigma^2_{\phi},\sigma^2_{\rho})$ where $\sigma_{\phi}= \sigma_{0,\phi} + \beta_\phi || x_p - u^s ||$, $\sigma_{\rho} = \sigma_{0,\rho}+ \beta_\rho || x_p - u^s||$. 
    \textbf{\textit{(ii)}} To demonstrate the sensor-agnostic nature of our approach, we consider agents equipped with a vision-based sensor. Each detected object $x$ leads to a measurement $z$ of noisy $xy$ positions, given by: $z = \big[p_x,p_y]^T + v$. Here, $v \sim \mathcal{N}(0,R)$ with $R = \mathrm{diag}(\sigma^2_{x},\sigma^2_{y})$.

    \textbf{Experiment 1~}\textit{Comparing greedy and brute force algorithm results for our submodular multi-objective value function}: 
    Figure~\ref{fig_reward_ratio} depicts the ratio of our multi-objective value function obtained from greedy and brute-force algorithms for the four scenarios. The result obtained from 20 Monte-Carlo (MC) runs for each scenario agrees with the performance guarantee of the greedy algorithm to yield an approximately optimal solution with a bounded performance guarantee at $(1-1/e)$ OPT.

    \textbf{Experiment 2:} \textit{Comparing multi-objective multi-agent planning with single objective multi-agent planning}. 
    Table~\ref{tab_scenario1234} compares results for scenario 1, 2, 3 and 4 
    collected from 20 MC runs for agents with \textit{range and bearing} based sensors. 
    It is expected that the average search area entropy is smallest for $V_2$ since it encourages agents to explore the search area.  Consequently, $V_2$ can also be seen to generate the best performance in term so of OSPA cardinality---\textbf{OSPA Card}. In contrast, we can see that the multi-agent planning with the single value function (to encourage only tracking accuracy) $V_1$, achieves improved results for object localization accuracy only (low \textbf{OSPA Loc} results) but at the expense of missing objects often out of the range of the sensors (as seen by significantly large \textbf{OSPA Card} results).  Most importantly, our results verify that $V_{mo}$ performs best in term of overall tracking and cardinality accuracy (reported by \textbf{OSPA Dist}) since $V_{mo}$ not only rewards agents for undertaking the discovery of new objects but also rewards agents for accurately tracking discovered objects. 

	Figure~\ref{fig_scenario3} shows the grid occupancy probability and the trajectories of the agents for scenario 3. 
	The results demonstrate the effectiveness of our proposed planning method, where agents not only track but discover distant mobile objects. 
    
    \begin{figure*}[!tb]
        \centering
        \includegraphics[width=1.6\columnwidth]{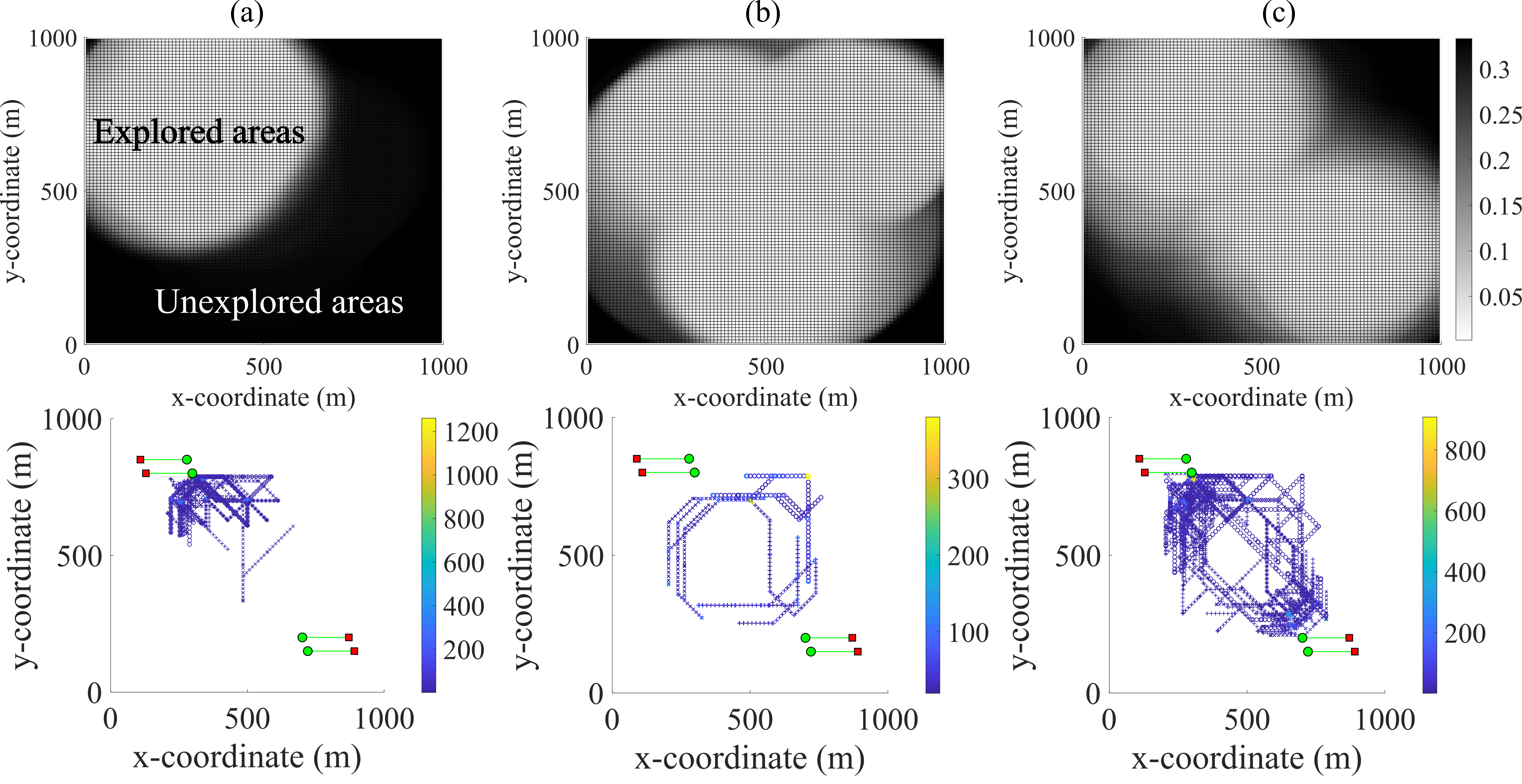}  
        \caption{\textbf{Scenario 3}. Grid occupancy probability (top) and heat map (bottom) of trajectories for $3$ agents over $20$ MC runs with $r_d=200$~m using (a) $V_1$. Late birth group $B$ never discovered, (b) $V_2$. Extensive exploration, and (c) $V_{mo}$. Discovers the late birth group $B$ whilst tracking both groups.}
        \label{fig_scenario3}
    \end{figure*}
    
    \begin{figure}[!tb]
        \centering
        \includegraphics[width=0.8\columnwidth]{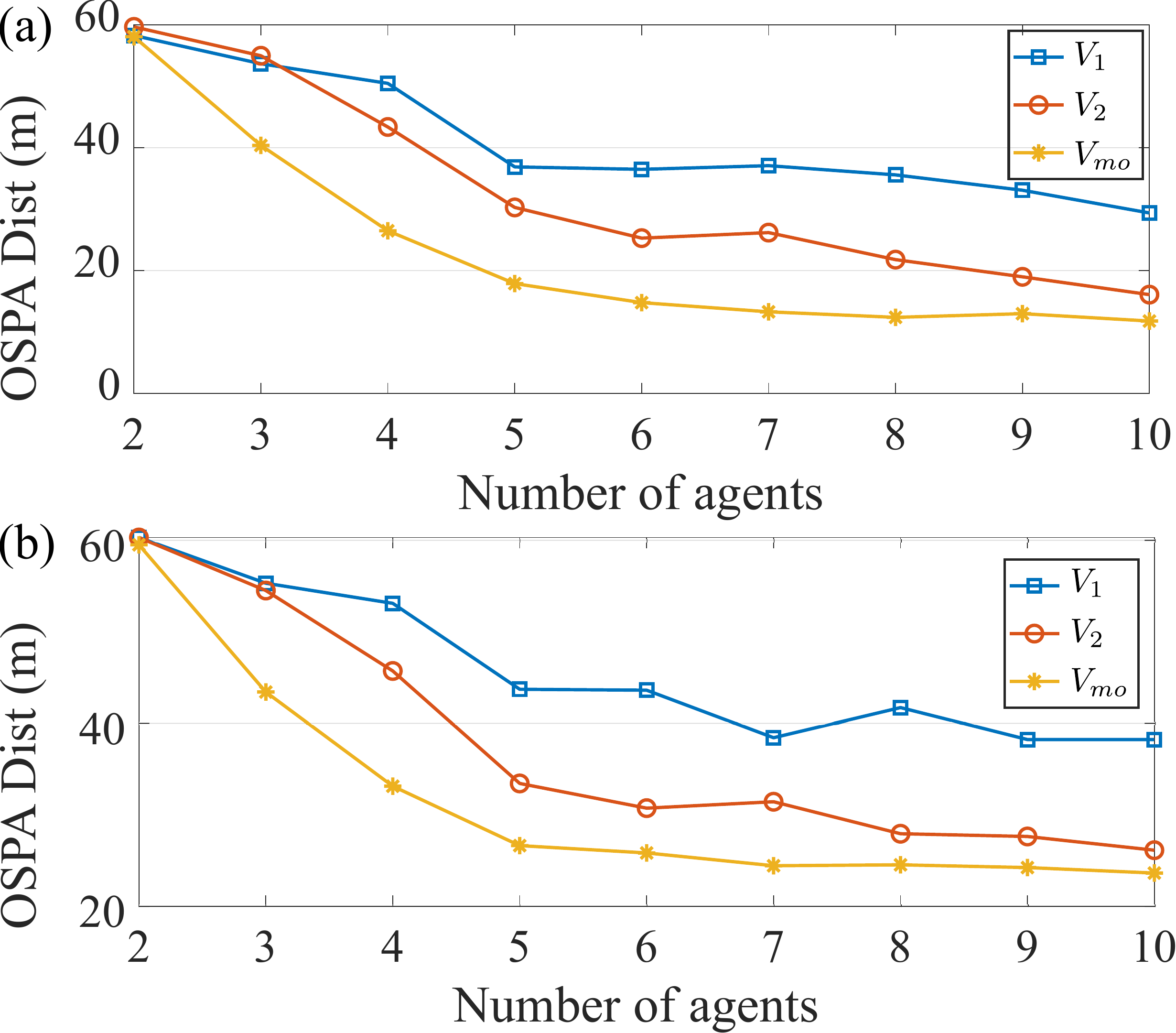}  
        \caption{Overall tracking performance over 20 MC runs based on multi-agent planning with our multi-objective value function $V_{mo}$ compare with the single objective value functions $V_1$ and $V_2$ when the number of agents are increased from $2$ to $10$ for Scenario 4 (\textbf{Explosion}) with $r_d = 200$~m using (a) agents with \textit{range and bearing} based sensors, (b) agents with \textit{vision} based sensors.}
        \label{fig_scenario4_ospa_dist_vs_ns}
    \end{figure}
    
    \textbf{Experiment 3:~}\textit{Explore the asymptotic behavior of tracking performance with an increasing number of agents for our planning formulation}.
    Figure~\ref{fig_scenario4_ospa_dist_vs_ns} depicts the overall mean tracking accuracy from 20 MC runs for agent teams with each detection-based sensor. It confirms that planning with $V_{mo}$ consistently performs better than $V_1$ or $V_2$ alone. As expected, when the number of agents increases, $V_1$ and $V_2$ tracking performances improve and approach that of $V_{mo}$. Interestingly, multi-agent planning with a single exploration objective closely approaches the tracking performance of the multi-objective value function when the team of agents is large enough to cover the survey area with its range limited sensors and all objects become visible to the agents.

	\section{Conclusion}
    In this paper, we have formulated a multi-objective planning approach for multi-agent tracking and searching for mobile objects. We have established that our formulation results in a value function that is monotone and submodular. We presented a series of extensive experimental results to demonstrate the effectiveness of our method and performance guarantees when using the low-cost greedy algorithm to determine control actions for the multi-agent. 
    
    We consider a centralized MPOMDP approach where scalability can be a limitation. Factored-POMDP \cite{oliehoek2008exploiting} can be employed to achieve further system scalability. We require reliable and fast communications between all agents and their centralized controller. If there are any delays in communications, the problem can be formulated as an MPOMDP with delayed communications. It is extremely challenging to plan and track mobile objects in an online manner without any communications among agents as in Dec-POMDP, hence it is an open question for future work. 
    
    \section{Acknowledgments}
	This work was jointly supported by the Western Australia Parks and Wildlife (WA Parks), the Australian Research Council (LP160101177,~DP160104662), the Defense Science and Technology Group (DSTG), and The Shultz Foundation. 
	
	\section{Appendix}

	\subsection{Parameter settings for experiments}
	The search areas for the first three scenarios and scenario 4 are $1000$~m $\times 1000$~m and $2000$~m $\times 2000$~m, respectively. Each agent is controlled to fly at a fixed and different altitude (\ie, $5$~m altitude gap between each agent) to prevent collisions with other team members. The minimum altitude starts at $30$~m for the first agent and increases $5$~m for each additional agent. Further, all objects are assumed exist on a horizontal ground plane 
	to speed up the numerical experiments by tracking in 2D. Each object state $\mathbf{x} = (x,l)$ is uniquely identified by its label $l$, while its motion state $x = [p_x,\dot{p}_x,p_y,\dot{p}_y]^T$ comprises of object's position and velocity in Cartesian coordinates. Each object moves in accordance with the constant velocity (CV) model given by $x_k = F^{CV}x_{k-1} +q^{CV}_{k-1}$. Here, $F^{CV} = [1,T_0;0,T_0] \otimes I_2$, $T_0$ is the sampling interval ($T_0 = 1$~s for our experiments), $\otimes$ denotes for the Kronecker tensor product; $I_2$ is the $2\times2$ identity matrix; $q^{CV}_{k-1} \sim \mathcal{N}(0,Q^{CV})$ is a $4\times 1$ zero mean Gaussian process noise, with co-variance $Q^{CV} = \sigma^2_{CV}[T^3_0/3,T^2_0/2;T^2_0/2,T_0]\otimes I_2$. The detection probability is $p_D(u_s,x_p) = $
	\begin{align}
	\begin{cases} 
	0.98 & ||x_p-u^s|| \leq r_{d} \\
	\max(0,0.98 - (||x_p-u^s|| - r_{d})\hbar) & \text{otherwise;}
	\end{cases} \notag
	\end{align} 
	where $r_d$ is the sensor detection range and $\hbar = 0.008$~m$^{-1}$. The sensor reports \textit{false detections} or false-alarm measurements following a Poison RFS with a clutter rate $\lambda = 0.2$, where each agent collects at most one measurement per time step for each object, either from the real objects, clutters (false detections) or the measurement is empty (missed detections). For sensor noise, the \textit{range and bearing} based measurement is corrupted with a zero mean Gaussian process noise that depends on the distance between objects and agents, \ie, $v \sim \mathcal{N}(0,R)$ with $R = \mathrm{diag}(\sigma^2_{\phi},\sigma^2_{\rho})$ where $\sigma_{\phi}= \sigma_{0,\phi} + \beta_\phi || x_p - u^s ||$, $\sigma_{\rho} = \sigma_{0,\rho}+ \beta_\rho || x_p - u^s||$; $\sigma_{0,\phi} = 2\pi/180$~rad, $\beta_{\phi} = 1.7 \cdot 10^{-5}$~rad/m, $\sigma_{0,\rho} = 10$~m, and $\beta_{\rho} = 5\cdot 10^{-3}$. Similarly, for \textit{vision-based} sensor, each detected object $x$ leads to a measurement $z$ of noisy $x-y$ positions, given by: $z = \big[p_x,p_y]^T + v$. Here, $v \sim \mathcal{N}(0,R)$ with $R = \mathrm{diag}(\sigma^2_{x},\sigma^2_{y})$ where $\sigma_{x}= \sigma_y =  \sigma_{0,xy} + \beta_{xy} || x_p - u^s ||$ with $\sigma_{0,xy} = 10$~m, and $\beta_{xy} = 1\cdot 10^{-2}$. 
	The grid size is $100 \times 100$ across four scenarios. This corresponds to a grid cell of $10$~m $\times 10$~m for scenario 1,2 and 3 and a grid cell of $20$~m $\times 20$~m for scenario 4. The total time is $200$~s. The agent does not have any prior knowledge about object's state, thus it uses the initial birth probability $r_B = 0.005$, and a Gaussian density $p_B = \mathcal{N}(x;m_B,Q_B)$ with $m_B = [500,0,500,0]^T$ and $Q_B = \mathrm{diag}([500,10,500,10])$.
	
		\subsection{Multi-agent POMDP (MPOMDP)}
	Multi-agent POMDP is a centralized control framework for multiple agents wherein each agent shares its observations via communications to a centralized controller. Let $\mathcal{F}(A)$ denote the class of finite subsert of $A$. An MPOMDP is described by a tuple $\Big[S, H, \mathcal{F}(\mathbb{X}) \times \mathbb{U}^S, \mathbb{A}^S,\mathcal{F}(\mathbb{Z})^S, \mathcal{T},\mathcal{R}, \mathcal{O} \Big]$ where
	\begin{itemize}
		\item $S$ is the number of agents.
		\item $H$ is the look-ahead horizon.
		\item $\mathcal{F}(\mathbb{X}) \times \mathbb{U}^S$ is the space, wherein each element is an ordered pair $(X,U)$, with $X$ is the object state and $U = [u^1,\dots,u^S]^T \in \mathbb{U}^S$ is states of $S$ agents. 
		\item $\mathbb{A}^S = \mathbb{A} \times \dots \times \mathbb{A}$ is the control action space for $S$ agents resulting in the joint action $A = [a_1,\dots,a^S]^T \in \mathbb{A}^S $
		\item $\mathcal{F}(\mathbb{Z})^S = \mathcal{F}(\mathbb{Z}) \times \dots \times \mathcal{F}(\mathbb{Z})$ is the space of joint observations resulting in the joint observation $Z = [Z^1,\dots,Z^S]^T \in \mathcal{F}(\mathbb{Z})^S$.
		\item $\mathcal{T}: \big[\mathcal{F}(\mathbb{X}) \times \mathbb{U}^S\big]  \times \big[\mathcal{F}(\mathbb{X}) \times \mathbb{U}^S\big] \times \mathbb{A}^S \rightarrow [0,1] $ defines the transition probabilities $\mathrm{Pr}\big((X',U')|(X,U),A\big)$.
		\item $\mathcal{R}: \big[\mathcal{F}(\mathbb{X}) \times \mathbb{U}^S\big]  \times \mathbb{A}^S \rightarrow \mathbb{R}$ defines the immediate reward of performing action $A$ in state $\big[\mathcal{F}(\mathbb{X}) \times \mathbb{U}^S\big]$.
		\item $\mathcal{O}: \mathcal{F}(\mathbb{Z})^S \times  \big[\mathcal{F}(\mathbb{X}) \times \mathbb{U}^S\big] \times \mathbb{A}^S \rightarrow [0,1] $ defines the joint observation probabilities $\mathrm{Pr}\big(Z|(X',U'),A\big)$.
	\end{itemize}
	\subsection{The Optimal Sub-Pattern Assignment (OSPA)}
	Let $\mathcal{F}(\mathbb{X})$ be the space of finite subsets of $\mathbb{X}$. According to \cite{schuhmacher2008consistent}, let $\bar{d}_p^{(c)}(X,Y)$ be the ~\textbf{OSPA Dist}  between $X,Y \in \mathcal{F}(\mathbb{X})$  with order $p$ and cutoff $c$. Let $m$ is the cardinality of $X$ and $n$ is cardinality of $Y$, with $m \leq n$, $\bar{d}_p^{(c)}(X,Y)$ is defined as 
	\begin{align}
	&\bar{d}_p^{(c)}(X,Y) \notag \\ &= \Bigg(\dfrac{1}{n} \bigg( \min_{\pi \in \Pi_n} \sum_{i=1}^m d^{(c)}(x^i,y^{\pi(i)})^p +c^p(n-m) \bigg) \Bigg)^{1/p}
	\end{align}
	where $d^{(c)}(x,y) = \min(c,d(x,y))$, in which $d(\cdot,\cdot)$ is a metric on the single object state space $\mathbb{X}$. If $m > n$, then $\bar{d}_p^{(c)}(X,Y) \triangleq \bar{d}_p^{(c)}(Y,X)$.
	
	The \textbf{OSPA Dist} is comprised of two components: \textbf{OSPA Loc} $\bar{d}_{p,\text{loc}}^{(c)}$ and \textbf{OSPA Card} $\bar{d}_{p,\text{card}}^{(c)}$ to account for localization and cardnality errors. These components are given by with $m \leq n$: 
	\begin{align}
	\bar{d}_{p,\text{loc}}^{(c)}(X,Y) &=  \bigg( \dfrac{1}{n}  \min_{\pi \in \Pi_n} \sum_{i=1}^m d^{(c)}(x^i,y^{\pi(i)})^p \bigg)^{1/p}, \\
	\bar{d}_{p,\text{card}}^{(c)}(X,Y) &= \bigg(\dfrac{c^p(n-m)}{n}\bigg)^{1/p}
	\end{align}
	and if $m >n$, then $\bar{d}_{p,\text{loc}}^{(c)}(X,Y) \triangleq \bar{d}_{p,\text{loc}}^{(c)}(Y,X)$, $\bar{d}_{p,\text{card}}^{(c)}(X,Y) \triangleq \bar{d}_{p,\text{card}}^{(c)}(Y,X)$.

\end{document}